\documentclass[a4paper]{article}
\usepackage{amssymb}
\usepackage{amsmath}
\usepackage{amsthm}

\def\CP{\mathcal{P}}

\newtheorem{theorem}{Theorem}
\newtheorem{protocol}{Protocol}

\newtheorem{definition}{Definition}
\newtheorem{lemma}{Lemma}
\newtheorem{prop}{Properties}

\def\bits{\{0,1\}}
\newcommand{\Prob}[1]{{\Pr\left[\,{#1}\,\right]}}

\begin{document}

\title{On the Oblivious Transfer Capacity of Generalized Erasure Channels against Malicious Adversaries}

\author{Rafael~Dowsley$^1$ and Anderson~C.~A.~Nascimento$^2$\\
\\
$^1$ Institute of Theoretical Informatics\\
Karlsruhe Institute of Technology\\ 
rafael.dowsley@kit.edu\\
\\
$^2$ Institute of Technology\\
University of Washington Tacoma\\
andclay@uw.edu.
}

\markboth{}%
{}
\maketitle

\begin{abstract} Noisy channels are a powerful resource for cryptography as they can be used to obtain
information-theoretically secure key agreement, commitment and oblivious transfer protocols, among others.
Oblivious transfer (OT) is a fundamental primitive since it is complete for secure multi-party computation, and 
the OT capacity characterizes how efficiently a channel can be used for obtaining 
string oblivious transfer. Ahlswede and Csisz\'{a}r (\emph{ISIT'07}) presented upper and lower bounds on the
OT capacity of generalized erasure channels (GEC) against passive adversaries. In the case of GEC 
with erasure probability at least 1/2, the upper and lower bounds match and therefore the OT capacity
was determined. It was later proved by Pinto et al. (\emph{IEEE Trans. Inf. Theory 57(8)}) that in this case there is 
also a protocol against malicious adversaries achieving the same lower bound, and hence the OT capacity is 
identical for passive and malicious adversaries. In the case of GEC with erasure probability smaller than 1/2, 
the known lower bound against passive adversaries that was established by Ahlswede and Csisz\'{a}r does not 
match their upper bound and it was unknown whether this OT rate could be achieved against malicious 
adversaries as well. In this work we show that there is a protocol against malicious adversaries achieving the same 
OT rate that was obtained against passive adversaries.

In order to obtain our results we introduce a novel use of interactive hashing that is suitable for dealing with the case of low erasure probability ($p^* <1/2$).
\end{abstract}

\maketitle

\section{Introduction} The usefulness of noisy channels for cryptographic purposes was first 
realized by Wyner \cite{BLTJ:Wyner75}, who proposed a secret key agreement protocol based on noisy channels. 
Later on it was showed by Cr\'{e}peau and Kilian that such channels can also be used to obtain information-theoretically secure 
implementations of cryptographic primitives such as oblivious transfer and commitment protocols~\cite{FOCS:CreKil88,EC:Crepeau97}.

Oblivious transfer (OT) is one of the fundamental cryptographic primitives since it is complete for two-party and multi-party 
computation~\cite{STOC:GolMicWig87,STOC:Kilian88,C:CreVanTap95}, i.e., given an implementation of OT it is possible to securely evaluate any polynomial 
time computable function without any additional assumptions. In the early years of research on OT, different variants of OT were
proposed~\cite{Wiesner83,TR:Rabin81}, but it was later showed that they are equivalent~\cite{C:Crepeau87}. Thereafter the community has 
focused mainly on the one-out-of-two string oblivious transfer variant, which is the one considered in this work. It is a primitive involving
two parts, Alice and Bob. Alice inputs two strings $S_{0},S_{1}\in \bits^{k}$ and Bob inputs a choice bit $c$. Bob receives as output 
$S_c$. The security of the OT protocol guarantees that (a dishonest) Alice cannot learn $c$, while (a dishonest) Bob cannot learn
both strings. The results of Cr\'{e}peau and Kilian~\cite{FOCS:CreKil88,EC:Crepeau97} regarding OT based on noisy channels 
were later improved in~\cite{KorMor01,ISIT:SteWol02,SCN:CreMorWol04,IEEEIT:NasWin08}.

\paragraph*{OT Capacity} After the initial success in obtaining OT protocols from noisy channels, researchers started to 
investigate the question of which channels can be used to implement OT and how efficiently this can be done. 
Nascimento and Winter~\cite{IEEEIT:NasWin08} proposed the notion of OT capacity, which is the 
optimal rate at which noisy channels can employed to realize OT, and also determined which noise resources 
have strictly positive OT capacity. Imai et al.~\cite{ISIT:ImaMorNas06} obtained the OT capacity of erasures channels 
against passive adversaries (i.e., adversaries which always follow the protocol) and a lower bound on its OT capacity against malicious adversaries 
(which can arbitrarily deviate from the protocol). Ahlswede and Csisz\'ar~\cite{ISIT:AhlCsi07,AhlCsi13} showed new bounds for the 
OT capacity of erasure channels.

\paragraph*{Generalized Erasure Channel} A generalized erasure channel (GEC) is a combination of a discrete memoryless 
channel and an erasure channel. The output of each transmission is an erasure with probability $p^*>0$,
independently from the input symbol. GECs represent a very special case for the study of OT based on noisy channels.
In fact, the known techniques to implement OT from noisy channels first use the noisy channel to emulate a GEC (in case that it is 
not already one) and then use the (emulated) GEC in the rest of the protocol. Thus, clarifying the OT capacity of the generalized
erasure channels is a central question.

Ahlswede and Csisz\'ar~\cite{ISIT:AhlCsi07,AhlCsi13} investigated the OT capacity of GECs against passive adversaries. 
For a GEC with $p^* \geq 1/2$, they determined the OT capacity. For a GEC with $p^* < 1/2$, 
they obtained upper and lower bounds for the OT capacity. Of course, the upper bounds also hold for the case of malicious adversaries.
Pinto et al.~\cite{IEEEIT:PDMN11} proved that for a GEC with $p^* \geq 1/2$, the OT rate achieved by
Ahlswede and Csisz\'ar's protocol against passive adversaries can also be achieved against malicious adversaries, and so the OT capacity
is the same. The techniques used in~\cite{IEEEIT:PDMN11} clearly do not apply in the case $p^*<1/2$ as they explicitly use the fact that the majority of the symbols received by Alice are erasures. 

\paragraph*{Our contribution} 

In this work we prove that for a GEC with $p^* < 1/2$, the same OT rate achieved by Ahlswede and Csisz\'ar's protocol~\cite{ISIT:AhlCsi07,AhlCsi13}
in the case of passive adversaries can also be achieved in the case of malicious adversaries, thus establishing a lower bound 
on the OT capacity of these GECs against malicious participants that is equal to one obtained against passive ones. We introduce a novel use of the interactive hashing techniques used by Cr\'{e}peau and Savvides in \cite{EC:CreSav06}.

\section{Preliminaries}\label{sec:pre}

\subsection{Notation} 
Domains of random variables and other sets will be denoted by calligraphic letters, 
the cardinality of a set $\mathcal{X}$ by $|\mathcal{X}|$, random variables by upper case letters, and 
realizations of the random variables by lower case letters. For a random variable $X$ over $\mathcal{X}$, 
$P_X: \mathcal{X} \rightarrow [0,1]$ with $\sum_{x \in \mathcal{X}} P_X(x) =1$ denotes its probability distribution.
For a joint probability distribution $P_{XY}: \mathcal{X} \times \mathcal{Y} \rightarrow [0,1]$, $P_X(x) := \sum_{y \in \mathcal{Y}}P_{XY}(x,y)$
denotes the marginal probability distribution and $P_{X|Y}(x|y):=\frac{P_{XY}(x,y)}{P_Y(y)}$ the conditional probability distribution
if $P_Y(y) \neq 0$. $X\in_R\mathcal{X}$ denotes a random variable uniformly distributed over $\mathcal{X}$ and $U_r$
a vector uniformly chosen from $\bits^r$. $[n]$ denotes the set $\{1,...,n\}$ and $[n] \choose \ell$ the 
set of all subsets $\mathcal{S} \subseteq [n]$, where $|\mathcal{S}| = \ell$. For $X^n=(X_1,X_2, \ldots, X_n)$ and 
$\mathcal{S} \subset [n]$, $X^{\mathcal{S}}$ is the restriction of $X^n$ to the positions in the subset $\mathcal{S}$.
Similarly for a set $\mathcal{R}$, $\mathcal{R}^{\mathcal{S}}$ is the subset of $\mathcal{R}$ consisting of 
the elements determined by $\mathcal{S}$. If $a$ and $b$ are two bit strings of the same dimension, $a \oplus b$ 
denotes their bitwise XOR. The logarithms used in this paper are in base 2. The entropy of $X$ is denoted by 
$H(X)$ and the mutual information between $X$ and $Y$ by $I(X;Y)$.

\subsection{Entropy and Extractors}

The main entropy measure used in this work is the min-entropy since its conditional version
captures the notion of unpredictability of a random variable, i.e., the private randomness 
that can be extracted from variable $X$ given the correlated random variable $Y$ possessed by an adversary.
For a finite alphabet $\mathcal{X}$, the min-entropy of a random variable $X \in \mathcal{X}$ is defined as
$$H_\infty(X)=\min\limits_{x}\log(1/P_{X}(x)).$$ Its conditional version, for a finite alphabet $\mathcal{Y}$ 
and a random variable $Y \in \mathcal{Y}$, is defined as
$$H_\infty(X|Y)=\min\limits_{y}H_\infty(X|Y=y).$$

For two probability distributions $P_X$ and $P_Y$ over the same domain $\mathcal{V}$, the statistical distance between
them is $$\mathsf{SD}(P_X,P_Y) := \frac{1}{2} \sum_{v \in \mathcal{V}} |P_X(v) - P_Y(v)|.$$

In order to extract secure one-time pads from random variables we use
strong extractors~\cite{NisZuc96,EC:DodReySmi04,DORS08}. 

\begin{definition}[Strong Extractors]
A probabilistic polynomial time function $\mathsf{Ext}: \{0,1\}^n \times \bits^r \rightarrow \{0,1\}^\ell$ 
using $r$ bits of randomness is a $(n,m,\ell,\epsilon)\mathrm{-strong}$  $\mathrm{extractor}$
if for all probability distributions $P_X$ with $\mathcal{X} = \{0,1\}^n$ and such 
that $H_\infty(X) \geq m$, we have that $\mathsf{SD}(P_{\mathsf{Ext}(X;U_r), U_r},P_{U_\ell,U_r}) \leq \epsilon$. 
\end{definition}

In particular we will use Universal Hash Functions~\cite{CarWeg79} as strong extractors since they can extract the 
optimal number of nearly random bits~\cite{RadTaS00} according to the Leftover-Hash Lemma (similarly the Privacy-Amplification Lemma)
~\cite{STOC:ImpLevLub89,HILL99,BBM88,IEEEIT:BBCM95}.

\begin{definition}[Universal Hash Function]
A class $\mathcal{G}$ of functions $g:\mathcal{X} \rightarrow \mathcal{Y}$ is \emph{2-universal} if, for any distinct $x_1,x_2 \in \mathcal{X}$, 
the probability that $g(x_1)=g(x_2)$ is at most $|\mathcal{Y}|^{-1}$ when $g$ is chosen uniformly at random from $\mathcal{G}$.
\end{definition} 

\begin{lemma} Let $\mathcal{G}$ be a 2-universal class of functions $g: \{0,1\}^n \rightarrow \{0,1\}^\ell$.
Then for $G$ chosen uniformly at random from $\mathcal{G}$ we have that
\[
\mathsf{SD}(P_{G(X),G},P_{U_\ell,G}) \leq \frac{1}{2} \sqrt{2^{-H_\infty(X)}2^\ell}.
\]
In particular, it is a $(n,m,\ell,\epsilon)\mathrm{-strong}$  $\mathrm{extractor}$ when $\ell \leq m-2\log(\epsilon^{-1})+2$.
\end{lemma}

\subsection{Interactive Hashing and Encoding of Subsets} \label{sec:IH}

The oblivious transfer protocol introduced in this paper uses interactive hashing as an important building block.
Interactive hashing is a cryptographic primitive between two players, the sender (Bob) and the receiver (Alice)
which was initially introduced in the context of computationally secure cryptography~\cite{OstVenYun93} but was
later on generalized for the context of information-theoretical cryptography. It is particularly useful in the 
design of unconditionally secure oblivious transfer protocols~\cite{FOCS:CacCreMar98,TCC:DHRS04,JC:DHRS07,
EC:CreSav06,IEEEIT:PDMN11,ISIT:DowLacNas14}. In this primitive Bob inputs a string $W\in\bits^m$ and both 
Alice and Bob receive as output two strings $W_0,W_1\in\bits^m$ such that $W_0 \neq W_1$. The first requirement 
is that one of the two output strings, $W_d$, should be equal to $W$. The second requirement is that one of the strings should be 
effectively beyond the control of (a malicious) Bob. On the other hand, the third requirement states that (a malicious) Alice should
not be able to learn $d$ (as long as $W_0$ and $W_1$ are a priori equally likely to be the input).

\begin{definition}[Security of Interactive Hashing~\cite{TCC:DHRS04,JC:DHRS07}]
An interactive hashing protocol is \emph{secure for Bob} if for every unbounded strategy of Alice 
($A'$), and every $W$, if $W_0$, $W_1$ are the outputs of the protocol between an honest Bob ($B$)
with input $W$ and $A'$, then $$\{View_{A'}^{\langle A', B\rangle}(W)|W=W_0\}=\{View_{A'}^{\langle A',B\rangle}(W)|W=W_1\},$$  
where $View_{A'}^{\langle A', B\rangle}(W)$ 
is Alice's view of the protocol when Bob's input is $W$. An interactive hashing protocol is 
($s, \rho$)-\emph{secure for Alice} if for every $S \subseteq \bits^m$ of size at most $2^s$ and every 
unbounded strategy of Bob ($B'$), if $W_0$, $W_1$ are the outputs of the protocol, then
\[
\mathrm{Pr}[W_0,W_1 \in S] < \rho,
\]
where the probability is taken over the coin tosses of Alice and Bob. An interactive hashing protocol is 
($s, \rho$)-\emph{secure} if it is secure for Bob and ($s, \rho$)-\emph{secure for Alice}.
\end{definition}

If the distribution of the string $W_{\bar{d}}$ over the randomness of the two parties is $\eta$-close 
to uniform on all strings not equal to $W_d$, then the protocol is called {\em $\eta$-uniform} 
interactive hashing.

\begin{lemma}[\cite{TCC:DHRS04,JC:DHRS07}]
\label{lem:interactive-hashing-ding}
Let $t, m$ be positive integers such that $t \geq \log m + 2$. Then there
exists a four-message $(2^{-m})$-uniform $(t,2^{-(m-t)+O(\log m)})$-secure 
interactive hashing protocol.
\end{lemma}

The interactive hashing scheme ensures that one of the outputs is almost 
uniformly random; however, in the oblivious transfer protocol, the two strings are not
used directly, but as encodings of subsets. For the protocol to succeed, both 
output strings should be valid encodings of subsets. Cover showed~\cite{IEEEIT:Cover73} the 
existence of an efficiently computable one to one mapping $F : {[n] \choose \ell}\rightarrow [{n \choose \ell}]$ 
for every integer $\ell\leq n$ (thus making it possible to encode the set $[n] \choose \ell$ in binary
strings of length $m = \lceil{ \log{n \choose \ell} }\rceil$). But using such mapping in a straight way
may result in only slightly more than half of the strings being valid encodings. Therefore 
we use the modified encoding of Savvides~\cite{Savvides07}, in which each string $W \in \bits^m$
encodes the same subset as $W \mod \binom{n}{\ell}$, thus implying that all strings always encode
valid subsets. In this encoding, each subset corresponds to either 1 or 2 strings in $\bits^m$, 
so this scheme can at most double the fraction of the strings that maps to Bob's subset of 
interest.

\section{Security Model}\label{sec:secmod}

In this section we specify the model used for proving the security of the
oblivious transfer protocol and also the resources available to the parties.
In the one-out-of-two string oblivious transfer, Alice gives two strings $S_0,S_1 \in \bits^k$ 
as input and Bob inputs a choice bit $c$.  Bob receives $S_c$ as output and remains 
ignorant about $S_{\overline c}$, while Alice should not learn Bob's choice bit.
As showed by Beaver~\cite{C:Beaver95}, there exists a very efficient reduction from
randomized OT to OT, therefore in this paper we consider for simplicity OT with random inputs. 
We consider malicious adversaries that can act arbitrarily. The protocol
participants are connected by both a noiseless channel and a generalized erasure channel. 
The security parameter $n$ determines the number of times that the generalized erasure channel
can be used.

 \begin{definition}[Generalized Erasure Channel~\cite{ISIT:AhlCsi07,AhlCsi13}]
A discrete memoryless channel $\{W: \mathcal{X} \rightarrow \mathcal{Y}\}$ is called
a \emph{generalized erasure channel} (GEC) if the output alphabet $\mathcal{Y}$ can be decomposed 
as $\mathcal{Y}_0 \cup \mathcal{Y}^*$ such that $W(y|x)$ does not depend on $x \in \mathcal{X}$, 
if $y \in \mathcal{Y}^*$. For a GEC, we denote $W_0(y|x)=\frac{1}{1-p^*}W(y|x)$, $x \in \mathcal{X}$, 
$y \in \mathcal{Y}_0$, where $p^*$ is the sum of $W(y|x)$ for $y \in \mathcal{Y}^*$ (not depending on $x$).
\end{definition}

We use the OT security definition from Cr\'{e}peau and Wullschleger~\cite{ICITS:CreWul08} because 
it implies the sequential composability of the protocols that meet it. Their definition is described below.
The statistical information of $X$ and $Y$ given $Z$ is defined as 
$$I_{Stat}(X;Y|Z)= \mathsf{SD} (P_{XYZ}, P_Z P_{X|Z} P_{Y|Z}).$$

A $\mathcal{F}$-hybrid protocol consists of a pair of algorithms $\CP = (A,B)$ that can interact 
and have access to some functionality $\mathcal{F}$. 
A pair of algorithms $\widetilde{P} = (\widetilde{A},\widetilde{B})$ is admissible for
protocol $\CP$ if at least one of the parties is honest, that is, if at least one of the equalities 
$\widetilde{A} = A$ and $\widetilde{B} = B$ holds. Let $S$ denote $(S_0,S_1)$.

\begin{theorem}[\cite{ICITS:CreWul08}]
\label{th:secOT} A protocol $\CP$ securely realizes string OT (for $k$-bit strings) with an error of 
at most $6 \epsilon$ if for every admissible pair of algorithms $\widetilde{P} = (\widetilde{A},\widetilde{B})$ 
for protocol $\CP$ and for all inputs $(S,C)$, $\widetilde{P}$ produces outputs $(U,V)$ such that the following
conditions are satisfied:%
\begin{itemize}
	\item (Correctness) If both parties are honest, then $U= \bot$ and 
$\Pr[V=S_C]\geq 1 - \epsilon$.
	\item (Security for Alice) If Alice is honest, then we have $U=\perp$ 
 and there exists a random variable $C'$ distributed according to $P_{C'|S,C,V}$, such that 
$I_{Stat}(S;C'|C) \leq \epsilon$ and $I_{Stat}(S;V|C,C',S_{C'}) \leq \epsilon$.
	\item (Security for Bob) If Bob is honest, we have $V \in \{0, 1\}^k$ and 
$I_{Stat}(C;U|S) \leq \epsilon$.
\end{itemize}
The protocol is secure if $\epsilon$ is negligible in the security parameter $n$.
\end{theorem}

If the protocol uses the generalized erasure channel $n$ times, its \emph{oblivious transfer rate} 
is given by $R_{OT}=\frac{k}{n}$. The \emph{oblivious transfer capacity}~\cite{IEEEIT:NasWin08} 
$C_{OT}$ is the supremum of the achievable rates with secure protocols.

\section{OT Capacity of GEC}\label{ot_gec}

For a generalized erasure channel $\{W: \mathcal{X} \rightarrow \mathcal{Y}\}$, 
let $C(W_0)$ denote the Shannon capacity of the 
discrete memoryless channel $\{W_0:\mathcal{X} \rightarrow \mathcal{Y}_0\}$.
For the case of generalized erasure channels with $p^* \geq \frac{1}{2}$, the oblivious transfer 
capacity was determined by Ahlswede and Csisz\'ar~\cite{ISIT:AhlCsi07,AhlCsi13} against passive
adversaries (i.e., adversaries that always follow the protocol) and Pinto 
et al.~\cite{IEEEIT:PDMN11} against malicious adversaries.

\begin{theorem}[\cite{ISIT:AhlCsi07,IEEEIT:PDMN11,AhlCsi13}]
For a generalized erasure channel with $p^* \geq \frac{1}{2}$, the oblivious transfer capacity 
both in the case of passive adversaries as in the case of malicious adversaries is 
$C_{OT}=(1-p^*) C(W_0)$.
\end{theorem}

For the case of generalized erasure channels with $p^* < \frac{1}{2}$, a lower bound on the
OT capacity against passive adversaries was obtained by Ahlswede and Csisz\'ar~\cite{ISIT:AhlCsi07,AhlCsi13}.

\begin{theorem}[\cite{ISIT:AhlCsi07,AhlCsi13}]
For a generalized erasure channel with $p^* < \frac{1}{2}$, a lower bound on the oblivious transfer 
capacity in the case of passive adversaries is $C_{OT} \geq p^* C(W_0)$.
\end{theorem}

In the current work we prove that the same OT rate that was achieved against passive adversaries 
can also be achieved against malicious ones.

\begin{theorem}\label{the:main}
For a generalized erasure channel with $p^* < \frac{1}{2}$, a lower bound on the oblivious transfer 
capacity in the case of \emph{malicious} adversaries is $C_{OT} \geq p^* C(W_0)$.
\end{theorem}

We present next a protocol that achieves such OT rate and its security proof. This protocol 
belongs to the lineage of OT protocols initiated by Cr\'{e}peau and Savvides~\cite{EC:CreSav06, Savvides07},
which use interactive hashing as a central, efficient mechanism to ensure that (a malicious) Bob 
is following the protocol rules without revealing to Alice his choice bit. Due to the fact that
in our case the non-erasure positions are the majority, our usage of the interactive hashing 
protocol is different from the previous protocols.

\begin{protocol}\label{prot:ot}
\mbox{}
{\rm 
\begin{enumerate}
\item (Parameter Setting) Alice and Bob select a positive constant 
$\alpha$ such that $3 \alpha <  1/2-p^*$ and set $\beta=1/2-p^*-\alpha$.
Note that $\beta > 2\alpha$.
\item (GEC Usage) Alice chooses $x^n$ randomly according to the probability
distribution that achieves the Shannon capacity of $W_0$. She sends
$x^n$ to Bob using the GEC, who receives the string $y^n$.
\item (Good/Bad Sets) Bob divides the string $y^n$
into a set $\mathcal{G}$ of good positions (those with $y \in \mathcal{Y}_0$) and 
a set $\mathcal{B}$ of bad positions (those corresponding to erasures).
The protocol is aborted if $|\mathcal{G}|<(1-p^*-\alpha)n$.
\item (Partitioning) Bob chooses uniformly randomly a bit $c$ and a $m$-bit string $w$, where 
$m= {\lceil{ \log{n/2 \choose \beta n} }\rceil}$. He decodes $w$ into a subset $\mathcal{T}$ of 
cardinality $\beta n$ out of $n/2$ (using the encoding scheme described in Section~\ref{sec:pre}). 
Then he partitions the $n$ positions into two sets of same cardinality. For $\mathcal{R}_c$ 
he picks randomly, and without repetition, $n/2$ positions from $\mathcal{G}$. For $\mathcal{R}_{\overline c}$,
he first picks the subset $\mathcal{R}_{\overline c}^{\mathcal{T}}$ randomly from the remaining positions 
from $\mathcal{G}$ and then fills the rest of $\mathcal{R}_{\overline c}$ randomly with the $n/2- \beta n$ still
unused positions. Bob sends the descriptions of $\mathcal{R}_0$ and $\mathcal{R}_{1}$ to Alice, who aborts 
if there is some repeated position.
\item (Interactive Hashing) Bob sends $w$ to Alice using the interactive hashing protocol. Let $w_0,w_1$ be the output strings,
$\mathcal{T}_0$, $\mathcal{T}_1$ the decoded subsets and $d$ be such that $w_d = w$.
\item (Checking the Partitioning) Bob announces $a=d\oplus c$, $y^{\mathcal{R}_0^{\mathcal{T}_{\overline a}}}$ and 
$y^{\mathcal{R}_1^{\mathcal{T}_a}}$. Alice verifies if $y^{\mathcal{R}_0^{\mathcal{T}_{\overline a}}}$ and 
$y^{\mathcal{R}_1^{\mathcal{T}_a}}$ are $2\epsilon$-jointly typical with her input on these positions for
the channel $\{W_0:\mathcal{X} \rightarrow \mathcal{Y}_0\}$ (see Appendix for the considered definitions
of typicality); aborting if this is not the case.
\item (Strings Transmission) Let $\mathcal{Q}_0=\mathcal{R}_0 \setminus \mathcal{T}_{\overline a}$,
 $\mathcal{Q}_1=\mathcal{R}_1 \setminus \mathcal{T}_a$ and $\mu = p^*+ \alpha$. 
Alice randomly chooses 2-universal hash functions 
$g_0,g_1: \mathcal{X}^{\mu n} \rightarrow \bits^{\mu n[H(X|Y \in \mathcal{Y}_0)+\epsilon]}$ 
(with $\epsilon > 0$ such that the output length is integer) and computes $g_0(x^{\mathcal{Q}_0})$ and $g_1(x^{\mathcal{Q}_1})$. 
In addition she also randomly chooses 2-universal hash functions 
$h_0,h_1: \mathcal{X}^{\mu n} \rightarrow \bits^{\delta n}$, 
where $\delta=(\mu - 5 \alpha)H(X)-\mu(H(X|Y \in \mathcal{Y}_0)+\epsilon) - \gamma$ and 
$\gamma >0$ is such that the output length is integer. Alice sends 
Bob $g_0(x^{\mathcal{Q}_0})$, $g_1(x^{\mathcal{Q}_1})$ and the descriptions of 
$g_0$, $g_1$, $h_0$, $h_1$. She outputs $S_0=h_0(x^{\mathcal{Q}_0})$ and $S_1=h_1(x^{\mathcal{Q}_1})$.
\item (Output) Bob computes all possible $\widetilde{x}^{\mathcal{Q}_c}$ that are jointly typical with 
$y^{\mathcal{Q}_c}$ and satisfy $g_c(\widetilde{x}^{\mathcal{Q}_c})=g_c(x^{\mathcal{Q}_c})$. If there exists 
exactly one such $\widetilde{x}^{\mathcal{Q}_c}$, then Bob outputs $S_c=h_c(\widetilde{x}^{\mathcal{Q}_c})$;
otherwise $S_c=0^{\delta n}$.

\end{enumerate}
} 
\end{protocol}

\begin{theorem}
This string oblivious transfer protocol is secure.
\end{theorem}

\paragraph{Correctness} If both Alice and Bob are honest, Bob will get the correct output value unless
he aborts in the Good/Bad Sets step or if the he does not recover exactly $\widetilde{x}^{\mathcal{Q}_c}=x^{\mathcal{Q}_c}$
in the Output step. But the probability that Bob has to abort in the Good/Bad Sets step is a negligible function of the 
security parameter $n$ due to the the Chernoff bound~\cite{Chernoff52}. Bob does not recover the correct 
$\widetilde{x}^{\mathcal{Q}_c}=x^{\mathcal{Q}_c}$ if either $x^{\mathcal{Q}_c}$ is not jointly typical with $y^{\mathcal{Q}_c}$ 
or if there exists another $\overline{x}^{\mathcal{Q}_c}$ that has $g_c(\overline{x}^{\mathcal{Q}_c})=g_c(x^{\mathcal{Q}_c})$ and
is jointly typical with $y^{\mathcal{Q}_c}$. The former case only occurs with negligible probability due to the 
definition of joint typicality. For the latter case, an upper bound on the number of $\overline{x}^{\mathcal{Q}_c}$ that
are jointly typical with $y^{\mathcal{Q}_c}$ is $2^{\mu n[H(X|Y \in \mathcal{Y}_0)+\epsilon']}$, for $0<\epsilon'<\epsilon$ 
and $n$ sufficiently large. Therefore according to the Leftover-Hash Lemma, for $n$ sufficiently large, with overwhelming probability
$g_c(\overline{x}^{\mathcal{Q}_c}) \neq g_c(x^{\mathcal{Q}_c})$ for all these other $\overline{x}^{\mathcal{Q}_c}$ that
are jointly typical with $y^{\mathcal{Q}_c}$. As all events that can result in Bob not obtaining the correct output only 
occur with negligible probability in $n$, the protocol is correct.

\paragraph{Security for Bob} In a generalized erasure channel, each input symbol $x$ is erased 
with the same probability $p^*$. Hence Alice has no knowledge about the erasures and thus 
from Alice's point of view the sets $(\mathcal{R}_0,\mathcal{R}_1)$ are independent from the choice bit $c$.
The only other point where the bit $c$ is used is to compute $a=d\oplus c$ in the Checking the Partitioning 
step. The interactive hashing protocol is $\eta'<2^{-m}$ uniform, which is negligible since 
$m= {\lceil{ \log{n/2 \choose \beta n} }\rceil} = O(n)$ by applying Stirling's approximation. Thus 
with overwhelming probability $w_{\bar{d}}$ is uniform in $\bits^m\setminus w$, and so 
Alice's views are identical for $d=0$ and $d=1$. Hence she gains no information about $d$ and
therefore about $c$. Note that in the Output step Bob does not abort, so Alice cannot use 
reaction attacks. Therefore with overwhelming probability Alice's view of the protocol is independent from $c$.

\paragraph{Security for Alice} The proof of security for Alice follows the lines of 
Savvides' proof~\cite[Section~5.1]{Savvides07}, but we use new variants of the 
supporting definitions and lemmas due to the fact that we use the interactive hashing
protocol in a different way.

\begin{definition}
Let $u(\mathcal{R})$ be the number of positions contained in $\mathcal{R}$ such that the corresponding output 
at this position was an erasure.
\end{definition}

\begin{definition}
$\mathcal{T}$ is called \emph{good} for $\mathcal{R}$ if $u(\mathcal{R}^{\mathcal{T}}) < \alpha n$, otherwise 
it is called \emph{bad} for $\mathcal{R}$.
\end{definition}

The proof is divided in two cases as follows: (i) both $u(\mathcal{R}_0), u(\mathcal{R}_1) \geq 2 \alpha n$, 
(ii) either $u(\mathcal{R}_0)$ or $u(\mathcal{R}_1)$ is less than $2 \alpha n$. 

\paragraph*{Case 1} For proving Alice's security in the first case we will need the following lemmas.

\begin{lemma}
Let $\mathcal{R}$ be a set of cardinality $n/2$ such that $u(\mathcal{R}) \geq 2 \alpha n$. The fraction $f$ of subsets 
$\mathcal{T} \subset \mathcal{R}$ of cardinality $\beta n$ that are good for $\mathcal{R}$ satisfies $f < (1-2\alpha)^{\alpha n}$.
\end{lemma} 

\begin{proof}
We prove that a subset $\mathcal{T}$ chosen uniformly at random will be good for $\mathcal{R}$ with probability
smaller than $(1-2\alpha)^{\alpha n}$ using the probabilistic method. One way of choosing $\mathcal{T}$ is by picking
sequentially at random, and without replacement, $\beta n$ positions out of the $n/2$ positions in $\mathcal{R}$.
For $1<i<\beta n$, the probability $p_i$ that the $i$-th chosen position is a non-erasure given that the subset $\mathcal{T}$ 
does not have enough erasure positions so far to be considered bad for $\mathcal{R}$ (i.e., less than $\alpha n$ erasures) is 
upper bounded by
$$
p_i < 1 - \frac{2 \alpha n-\alpha n}{n/2} = 1 - 2 \alpha
$$
Since for a subset $\mathcal{T}$ to be considered good for $\mathcal{R}$ it needs to have at least $\beta n - \alpha n$ non-erasure
positions, we have that 
$$\Prob{\mathcal{T} \text{ is good for } \mathcal{R}} < (1 - 2 \alpha)^{\beta n - \alpha n} < (1 - 2 \alpha)^{\alpha n}.$$
\end{proof}

\begin{lemma}
Let $\mathcal{R}_0,\mathcal{R}_1$ be sets of cardinality $n/2$ such that $u(\mathcal{R}_0) \geq 2 \alpha n$ and 
$u(\mathcal{R}_1) \geq 2 \alpha n$. The fraction of strings $w$ that decode to subsets $\mathcal{T}$ that are good 
for either $\mathcal{R}_0$ or $\mathcal{R}_1$ is no larger than $4(1 - 2 \alpha)^{\alpha n}$.
\end{lemma}

\begin{proof}
It follows from the previous lemma and the union bound that the fraction $f$ of subsets $\mathcal{T}$ that are good
for either $\mathcal{R}_0$ or $\mathcal{R}_1$ is smaller than $2(1 - 2 \alpha)^{\alpha n}$. Then the lemma follows 
straightforwardly from the fact that in the encoding scheme there are either one or two strings mapping to each set.
\end{proof}

Since the fraction of the strings $w \in \bits^m$ that are good for either $\mathcal{R}_0$ or $\mathcal{R}_1$ is no 
larger than $4(1 - 2 \alpha)^{\alpha n}$, we can set the security parameter $s$ of the interactive hashing protocol to 
$\log (4(1 - 2 \alpha)^{\alpha n}2^m)=m+ \alpha n \log (1 - 2 \alpha) + 2$ and thus have
$\rho= 2^{-(m-s)+O(\log m)} = 2^{\alpha n \log (1 - 2 \alpha)+   O(\log n)}$. Hence, by the security of the interactive 
hashing protocol, the probability that both $w_0$ and $w_1$ are good for either $\mathcal{R}_0$ or $\mathcal{R}_1$
is a negligible function of $n$, and so with overwhelming probability one of the sets (w.l.o.g. $\mathcal{R}_0$) will 
have $u(\mathcal{R}_0^{\mathcal{T}_{\overline a}}) \geq \alpha n$. 
 
By lemma~\ref{lem:sub} (in the appendix), if two $n$ long strings are not jointly typical at a uniformly randomly chosen 
linear fraction of positions, then these $n$ long strings are not jointly typical. Hence Bob can only successfully pass the test
performed by Alice in the Checking the Partitioning step (i.e., he can only find $y^{\mathcal{R}_0^{\mathcal{T}_{\overline a}}}$ 
that is jointly typical with Alice's input) if he can correctly guess $y$'s values for the erasure positions that are jointly typical with
Alice's input on these positions. For $\epsilon > 0$ and $n$ sufficiently large, there are for these positions at most 
$2^{\alpha n [H(Y \in \mathcal{Y}_0|X)+ \epsilon]}$ sequences of $y$'s values that are jointly typical with Alice's input, 
and there are at least $2^{\alpha n [H(Y \in \mathcal{Y}_0)- \epsilon]}$ typical sequences for the $y$'s values, thus
Bob's success probability is less than $2^{\alpha n [H(Y \in \mathcal{Y}_0|X)- H(Y \in \mathcal{Y}_0)+ 2\epsilon]}=2^{-\alpha n [C(W_0)-2 \epsilon]}$, 
which is a negligible function of $n$. Since Bob can only cheat with negligible probability in the case that 
both $u(\mathcal{R}_0), u(\mathcal{R}_1) \geq 2 \alpha n$, the protocol is secure for Alice in this case.

\paragraph*{Case 2} We assume w.l.o.g. that $\mathcal{R}_0$ is the one with $u(\mathcal{R}_0)<2 \alpha n$. 
The Chernoff bound guarantees that $|\mathcal{B}|>(p^*-\alpha)n$ with overwhelming probability. If 
$\mathcal{T}_a$ is bad for $\mathcal{R}_{1}$, then, by the same reasons as above, we have 
that Bob can only successfully pass the test performed by Alice in the Checking the Partitioning step (i.e., 
finding $y^{\mathcal{R}_1^{\mathcal{T}_a}}$ that is jointly typical with Alice's input) with negligible probability. 
But if $u(\mathcal{R}_0)<2 \alpha n$, $u(\mathcal{R}_{1}^{\mathcal{T}_a}) < \alpha n$ and $|\mathcal{B}|>(p^*-\alpha)n$,
then $u(\mathcal{Q}_1) \geq (p^*-4\alpha)n$. Then from Bob's point of view, at least $(p^*-4\alpha)n=(\mu - 5 \alpha)n$ of the positions
in $\mathcal{Q}_1$ are erasures and Alice only sends him $\mu n[H(X|Y \in \mathcal{Y}_0)+\epsilon]$ bits of information 
about $x^{\mathcal{Q}_1}$. Hence 
$H_\infty(X^{\mathcal{Q}_1}|\mathrm{View_{Bob}})>n[(\mu - 5 \alpha)H(X)-\mu H(X|Y \in \mathcal{Y}_0) - \mu \epsilon]$ and 
so the use of the 2-universal hash function $h_1$ for extracting $n[(\mu - 5 \alpha)H(X)-\mu H(X|Y \in \mathcal{Y}_0) - \mu \epsilon -  \gamma]$
bits is secure according to the Leftover-Hash Lemma. Therefore the protocol is secure for Alice in this case as well.

\paragraph*{Maximizing the oblivious transfer rate} For $n$ sufficiently large, $\alpha$, $\epsilon$ and $\gamma$ 
can be made arbitrarily small without compromising the security, thus in the limit the 
strings' length can be up to $n p^* [H(X)-H(X|Y \in \mathcal{Y}_0)]$. Since the probability distribution
used for $X$ is the one achieving the Shannon capacity of $W_0$, this is equal to $n p^* C(W_0)$, thus 
proving Theorem~\ref{the:main}.

\section{Conclusions}

In this work it was proven that the known lower bound in case of passive adversaries for the oblivious transfer capacity 
of the generalized erasure channels with error probability $p^* < 1/2$  also holds in the case of malicious 
adversaries, which can deviate arbitrarily from the protocol. In order to prove this result, a novel usage of the interactive hashing 
technique suitable for channels with low erasure probability was established, which can be of interest in other 
scenarios. The question of determining the exact oblivious transfer capacity of the generalized erasure 
channels with low erasure probability remains open, even for passive adversaries, and would be an interesting direction 
for future research given the pivotal role of these channels in the known constructions of oblivious transfer from noisy channels. 
Another interesting line of research would be developing new methodologies for obtaining oblivious transfer 
from noisy channels which circumvent the need of emulating a generalized erasure channel as a first step.

\bibliographystyle{plain}

\begin{thebibliography}{10}

\bibitem{ISIT:AhlCsi07}
R.~Ahlswede and I~Csisz\'ar.
\newblock On oblivious transfer capacity.
\newblock In {\em Information Theory, 2007. ISIT 2007. IEEE International
  Symposium on}, pages 2061--2064, June 2007.

\bibitem{AhlCsi13}
Rudolf Ahlswede and Imre Csisz\'ar.
\newblock On oblivious transfer capacity.
\newblock In Harout Aydinian, Ferdinando Cicalese, and Christian Deppe,
  editors, {\em Information Theory, Combinatorics, and Search Theory}, volume
  7777 of {\em Lecture Notes in Computer Science}, pages 145--166. Springer
  Berlin Heidelberg, 2013.

\bibitem{C:Beaver95}
Donald Beaver.
\newblock Precomputing oblivious transfer.
\newblock In Don Coppersmith, editor, {\em Advances in Cryptology --
  {CRYPTO}'95}, volume 963 of {\em Lecture Notes in Computer Science}, pages
  97--109, Santa Barbara, CA, USA, August~27--31, 1995. Springer, Berlin,
  Germany.

\bibitem{IEEEIT:BBCM95}
C.H. Bennett, G.~Brassard, C.~Cr\'{e}peau, and U.M. Maurer.
\newblock Generalized privacy amplification.
\newblock {\em Information Theory, IEEE Transactions on}, 41(6):1915--1923, Nov
  1995.

\bibitem{BBM88}
Charles~H. Bennett, Gilles Brassard, and Jean-Marc Robert.
\newblock Privacy amplification by public discussion.
\newblock {\em SIAM J. Comput.}, 17(2):210--229, April 1988.

\bibitem{FOCS:CacCreMar98}
Christian Cachin, Claude Cr{\'e}peau, and Julien Marcil.
\newblock Oblivious transfer with a memory-bounded receiver.
\newblock In {\em 39th Annual Symposium on Foundations of Computer Science},
  pages 493--502, Palo Alto, California, USA, November~8--11, 1998. {IEEE}
  Computer Society Press.

\bibitem{CarWeg79}
J.~Lawrence Carter and Mark~N. Wegman.
\newblock Universal classes of hash functions.
\newblock {\em Journal of Computer and System Sciences}, 18(2):143 -- 154,
  1979.

\bibitem{Chernoff52}
Herman Chernoff.
\newblock A measure of asymptotic efficiency for tests of a hypothesis based on
  the sum of observations.
\newblock {\em The Annals of Mathematical Statistics}, 23(4):493--507, 12 1952.

\bibitem{IEEEIT:Cover73}
T.M. Cover.
\newblock Enumerative source encoding.
\newblock {\em Information Theory, IEEE Transactions on}, 19(1):73--77, Jan
  1973.

\bibitem{C:Crepeau87}
Claude Cr{\'e}peau.
\newblock Equivalence between two flavours of oblivious transfers.
\newblock In Carl Pomerance, editor, {\em Advances in Cryptology --
  {CRYPTO}'87}, volume 293 of {\em Lecture Notes in Computer Science}, pages
  350--354, Santa Barbara, CA, USA, August~16--20, 1987. Springer, Berlin,
  Germany.

\bibitem{EC:Crepeau97}
Claude Cr{\'e}peau.
\newblock Efficient cryptographic protocols based on noisy channels.
\newblock In Walter Fumy, editor, {\em Advances in Cryptology --
  {EUROCRYPT}'97}, volume 1233 of {\em Lecture Notes in Computer Science},
  pages 306--317, Konstanz, Germany, May~11--15, 1997. Springer, Berlin,
  Germany.

\bibitem{FOCS:CreKil88}
Claude Cr{\'e}peau and Joe Kilian.
\newblock Achieving oblivious transfer using weakened security assumptions
  (extended abstract).
\newblock In {\em 29th Annual Symposium on Foundations of Computer Science},
  pages 42--52, White Plains, New York, October~24--26, 1988. {IEEE} Computer
  Society Press.

\bibitem{SCN:CreMorWol04}
Claude Cr{\'e}peau, Kirill Morozov, and Stefan Wolf.
\newblock Efficient unconditional oblivious transfer from almost any noisy
  channel.
\newblock In Carlo Blundo and Stelvio Cimato, editors, {\em SCN 04: 4th
  International Conference on Security in Communication Networks}, volume 3352
  of {\em Lecture Notes in Computer Science}, pages 47--59, Amalfi, Italy,
  September~8--10, 2004. Springer, Berlin, Germany.

\bibitem{EC:CreSav06}
Claude Cr{\'e}peau and George Savvides.
\newblock Optimal reductions between oblivious transfers using interactive
  hashing.
\newblock In Serge Vaudenay, editor, {\em Advances in Cryptology --
  {EUROCRYPT}~2006}, volume 4004 of {\em Lecture Notes in Computer Science},
  pages 201--221, St. Petersburg, Russia, May~28~--~June~1, 2006. Springer,
  Berlin, Germany.

\bibitem{C:CreVanTap95}
Claude Cr{\'e}peau, Jeroen {van de Graaf}, and Alain Tapp.
\newblock Committed oblivious transfer and private multi-party computation.
\newblock In Don Coppersmith, editor, {\em Advances in Cryptology --
  {CRYPTO}'95}, volume 963 of {\em Lecture Notes in Computer Science}, pages
  110--123, Santa Barbara, CA, USA, August~27--31, 1995. Springer, Berlin,
  Germany.

\bibitem{ICITS:CreWul08}
Claude Cr\'{e}peau and J\"urg Wullschleger.
\newblock Statistical security conditions for two-party secure function
  evaluation.
\newblock In Reihaneh Safavi-Naini, editor, {\em Information Theoretic
  Security}, volume 5155 of {\em Lecture Notes in Computer Science}, pages
  86--99. Springer Berlin Heidelberg, 2008.

\bibitem{CsiKor82}
Imre Csisz\'ar and J\'anos K\"{o}rner.
\newblock {\em Information Theory: Coding Theorems for Discrete Memoryless
  Systems}.
\newblock Academic Press, Inc., Orlando, FL, USA, 1982.

\bibitem{TCC:DHRS04}
Yan~Zong Ding, Danny Harnik, Alon Rosen, and Ronen Shaltiel.
\newblock Constant-round oblivious transfer in the bounded storage model.
\newblock In Moni Naor, editor, {\em TCC~2004: 1st Theory of Cryptography
  Conference}, volume 2951 of {\em Lecture Notes in Computer Science}, pages
  446--472, Cambridge, MA, USA, February~19--21, 2004. Springer, Berlin,
  Germany.

\bibitem{JC:DHRS07}
Yan~Zong Ding, Danny Harnik, Alon Rosen, and Ronen Shaltiel.
\newblock Constant-round oblivious transfer in the bounded storage model.
\newblock {\em Journal of Cryptology}, 20(2):165--202, April 2007.

\bibitem{DORS08}
Yevgeniy Dodis, Rafail Ostrovsky, Leonid Reyzin, and Adam Smith.
\newblock Fuzzy extractors: How to generate strong keys from biometrics and
  other noisy data.
\newblock {\em SIAM J. Comput.}, 38(1):97--139, March 2008.

\bibitem{EC:DodReySmi04}
Yevgeniy Dodis, Leonid Reyzin, and Adam Smith.
\newblock Fuzzy extractors: How to generate strong keys from biometrics and
  other noisy data.
\newblock In Christian Cachin and Jan Camenisch, editors, {\em Advances in
  Cryptology -- {EUROCRYPT}~2004}, volume 3027 of {\em Lecture Notes in
  Computer Science}, pages 523--540, Interlaken, Switzerland, May~2--6, 2004.
  Springer, Berlin, Germany.

\bibitem{ISIT:DowLacNas14}
Rafael Dowsley, Felipe Lacerda, and Anderson~C.A. Nascimento.
\newblock Oblivious transfer in the bounded storage model with errors.
\newblock In {\em Information Theory, 2014 IEEE International Symposium on},
  pages 1623--1627, June 2014.

\bibitem{STOC:GolMicWig87}
Oded Goldreich, Silvio Micali, and Avi Wigderson.
\newblock How to play any mental game or {A} completeness theorem for protocols
  with honest majority.
\newblock In Alfred Aho, editor, {\em 19th Annual {ACM} Symposium on Theory of
  Computing}, pages 218--229, New York City,, New York, USA, May~25--27, 1987.
  {ACM} Press.

\bibitem{HILL99}
J.~H\r{a}stad, R.~Impagliazzo, L.~Levin, and M.~Luby.
\newblock A pseudorandom generator from any one-way function.
\newblock {\em SIAM Journal on Computing}, 28(4):1364--1396, 1999.

\bibitem{ISIT:ImaMorNas06}
H.~Imai, K.~Morozov, and A.C.A. Nascimento.
\newblock On the oblivious transfer capacity of the erasure channel.
\newblock In {\em Information Theory, 2006 IEEE International Symposium on},
  pages 1428--1431, July 2006.

\bibitem{STOC:ImpLevLub89}
Russell Impagliazzo, Leonid~A. Levin, and Michael Luby.
\newblock Pseudo-random generation from one-way functions (extended abstracts).
\newblock In {\em 21st Annual {ACM} Symposium on Theory of Computing}, pages
  12--24, Seattle, Washington, USA, May~15--17, 1989. {ACM} Press.

\bibitem{STOC:Kilian88}
Joe Kilian.
\newblock Founding cryptography on oblivious transfer.
\newblock In {\em 20th Annual {ACM} Symposium on Theory of Computing}, pages
  20--31, Chicago, Illinois, USA, May~2--4, 1988. {ACM} Press.

\bibitem{KorMor01}
Valeri Korjik and Kirill Morozov.
\newblock Generalized oblivious transfer protocols based on noisy channels.
\newblock In {\em Proceedings of the International Workshop on Information
  Assurance in Computer Networks: Methods, Models, and Architectures for
  Network Security}, MMM-ACNS '01, pages 219--229, London, UK, UK, 2001.
  Springer-Verlag.

\bibitem{IEEEIT:NasWin08}
A.C.A. Nascimento and A~Winter.
\newblock On the oblivious-transfer capacity of noisy resources.
\newblock {\em Information Theory, IEEE Transactions on}, 54(6):2572--2581,
  June 2008.

\bibitem{NisZuc96}
Noam Nisan and David Zuckerman.
\newblock Randomness is linear in space.
\newblock {\em J. Comput. Syst. Sci.}, 52(1):43--52, February 1996.

\bibitem{OstVenYun93}
Rafail Ostrovsky, Ramarathnam Venkatesan, and Moti Yung.
\newblock Fair games against an all-powerful adversary.
\newblock In Renato Capocelli, Alfredo De~Santis, and Ugo Vaccaro, editors,
  {\em Sequences II}, pages 418--429. Springer New York, 1993.

\bibitem{IEEEIT:PDMN11}
A.C.B. Pinto, R.~Dowsley, K.~Morozov, and A.C.A. Nascimento.
\newblock Achieving oblivious transfer capacity of generalized erasure channels
  in the malicious model.
\newblock {\em Information Theory, IEEE Transactions on}, 57(8):5566--5571, Aug
  2011.

\bibitem{TR:Rabin81}
Michael~O. Rabin.
\newblock How to exchange secrets by oblivious transfer.
\newblock Technical Report Technical Memo TR-81, Aiken Computation Laboratory,
  Harvard University, 1981.

\bibitem{RadTaS00}
J.~Radhakrishnan and A.~Ta-Shma.
\newblock Bounds for dispersers, extractors, and depth-two superconcentrators.
\newblock {\em SIAM Journal on Discrete Mathematics}, 13(1):2--24, 2000.

\bibitem{Savvides07}
George Savvides.
\newblock {\em Interactive Hashing and Reductions Between Oblivious Transfer
  Variants}.
\newblock PhD thesis, Montreal, Que., Canada, Canada, 2007.
\newblock AAINR32237.

\bibitem{ISIT:SteWol02}
D.~Stebila and S.~Wolf.
\newblock Efficient oblivious transfer from any non-trivial binary-symmetric
  channel.
\newblock In {\em Information Theory, 2002. Proceedings. 2002 IEEE
  International Symposium on}, pages 293--, 2002.

\bibitem{Wiesner83}
Stephen Wiesner.
\newblock Conjugate coding.
\newblock {\em SIGACT News}, 15(1):78--88, January 1983.

\bibitem{BLTJ:Wyner75}
A.~D. Wyner.
\newblock The wire-tap channel.
\newblock {\em Bell System Technical Journal}, 54(8):1355--1387, 1975.

\end{thebibliography}

\appendix

\section {Typical Sequences}

The following definitions follow largely the book of Csisz\'{a}r and K\"{o}rner \cite{CsiKor82}.

\begin{definition}For a probability distribution $P$ on $\mathcal{X}$ and $\epsilon > 0$ the \emph{$\epsilon$-typical sequences} form the set
\begin{eqnarray*}
\mathcal{T}_{P,\epsilon}^{n} = \{ x^n \in \mathcal{X}^n : \forall x \in \mathcal{X}~ |N(x|x^n) - nP(x)| \leq \epsilon n~\&~\\
P(x)=0 \Rightarrow N(x|x^n)=0 \},
\end{eqnarray*}
with the number $N(x|x^n)$ denoting the number of symbols $x$ in the string $x^n$. 
\end{definition}

The \emph{type} of $x^n$ is the probability distribution $P_{x^n}(x) = \frac{1}{n}N(x|x^n)$. Then, $x^n \in \mathcal{T}_{P,\epsilon}^{n} \Rightarrow |P_{x^n}(x) - P(x)| \leq  \epsilon, \forall x\in \mathcal{X}$.

\begin{prop}
\mbox{}
	\begin{enumerate}
	\item $P^{\otimes n}(\mathcal{T}_{P,\epsilon}^{n})\geq 1 - 2|\mathcal{X}| \exp(-n\epsilon^2/2)$.
	\item $ |\mathcal{T}_{P,\epsilon}^{n}| \leq  \exp(nH(P)+n\epsilon D) $.
	\item $ |\mathcal{T}_{P,\epsilon}^{n}| \geq  (1 - 2|\mathcal{X}| \exp(-n\epsilon^2/2)) \exp(nH(P)-n\epsilon D) $,
	\end{enumerate}
	with the constant $D = \sum_{x:P(x)\neq 0} -\log P(x)$. See \cite{CsiKor82} for more details.
\end{prop}

Extending this concept to the conditional $\epsilon$-typical sequences, we have:
\begin{definition}\label{cts}
Consider a channel $W: \mathcal{X}\rightarrow \mathcal{Y}$ and an input string $x \in \mathcal{X}^n$. For $\epsilon > 0$, the 
\emph{conditional $\epsilon$-typical sequences} form the set
\begin{eqnarray*}
\mathcal{T}_{W,\epsilon}^{n}(x^n) & = & \{ y^n : \forall x \in \mathcal{X},y \in \mathcal{Y}~ |N(xy|x^n y^n)\\
&  &-nW(y|x)P_{x^n}(x)| \leq \epsilon n \\
&  & \& \, W(y|x)=0 \Rightarrow N(xy|x^n y^n)=0 \} \\
 & = & \prod_{x} \mathcal{T}_{W_x,\epsilon P_{x^n}(x)^{-1},}^{\mathcal{I}_x}
\end{eqnarray*}
 where $\mathcal{I}_x$ are the sets  of positions in the string $x^n$ where $x_k = x$.
\end{definition}
 
\begin{prop}\label{cts_prop}
\mbox{}
	\begin{enumerate}
	\item \label{pr} $W_{x^n}^{n}(\mathcal{T}_{W,\epsilon}^{n})\geq 1 - 2|\mathcal{X}||\mathcal{Y}| \exp(-n\epsilon^2/2)$.
	\item $|\mathcal{T}_{W,\epsilon}^{n}| \leq \exp(nH(W|P_{x^n})+n\epsilon E)$.
	\item $|\mathcal{T}_{W,\epsilon}^{n}| \geq \left(1 - 2|\mathcal{X}||\mathcal{Y}| \exp(-n\epsilon^2/2) \right)$\\ $\cdot\exp(-nH(W|P_{x^n})-n\epsilon E)$,
	\end{enumerate}
	with the constant $E = \mbox{max}_x \sum_{y:W(y)\neq 0} -\log W_x(y)$ and the conditional entropy $H(W|P) = \sum_x P(x)H(W_x)$. See \cite{CsiKor82} for more details.
\end{prop}

It is a well know fact that if $x^n$ and $y^n$ are conditional $\epsilon$-typical according the definition \ref{cts}, then 
\[
|\mathcal{T}_{W,\epsilon}^{n}| \leq  2^{n(H(Y|X)+\epsilon)}.
\]

We now prove the following lemma:

\begin{lemma}\label{lem:sub}
Let $W: \mathcal{X}\rightarrow \mathcal{Y}$ be a discrete memoryless channel and $x^n \in \mathcal{X}^n$, $y^n \in \mathcal{Y}^n$ be the input and output strings of this channel. Let $\mathcal{A}$ be a random subset of $[n]$ such that $|\mathcal{A}| = \delta n$, $0< \delta \leq 1$. Let $x^{\mathcal{A}}$ and $y^{\mathcal{A}}$ be the restrictions of $x^n$ and $y^n$ to the positions in the set $\mathcal{A}$. If $x^n$ and $y^n$ are conditional $\epsilon$-typical, then $x^{\mathcal{A}}$ and $y^{\mathcal{A}}$ are conditional $2 \epsilon$-typical for any $\epsilon>0$ and $n$ large enough.
\end{lemma}

\begin{proof} 

By hypothesis $x^n$ and $y^n$ are conditional $\epsilon$-typical, so for every symbols $x$ and $y$ we have that
$$
\left|N(xy|x^{n}y^{n}) - n P_{x^{n}}(x) W(y|x)\right| \leq \epsilon n,
$$
for a large enough $n$.

Given the conditional $\epsilon$-typical strings $x^n$ and $y^n$, the probability of selecting one pair with the specific values $x$ and $y$ for the substrings $x^{\mathcal{A}}$ and $y^{\mathcal{A}}$ is $\frac{N(xy|x^{n}y^{n})}{n} $. We have that
$$
P_{x^{n}}(x) W(y|x)-\epsilon \leq \frac{N(xy|x^{n}y^{n})}{n} \leq P_{x^{n}}(x) W(y|x)+\epsilon.
$$
Therefore, by the Chernoff bound~\cite{Chernoff52}, for $n$ large enough with overwhelming probability the number of pairs of $x$ and $y$ in the substrings $x^{\mathcal{A}}$ and $y^{\mathcal{A}}$, $N(xy|x^{\mathcal{A}}y^{\mathcal{A}})$, is limited by
\begin{eqnarray*}
\delta n \left(P_{x^{n}}(x) W(y|x)-\epsilon-\epsilon'\right) \leq N(xy|x^{\mathcal{A}}y^{\mathcal{A}}) &~&\\
\leq \delta n \left(P_{x^{n}}(x) W(y|x)+\epsilon+\epsilon'\right),
\end{eqnarray*}
for any $\epsilon' > 0$. Making $\epsilon ' = \epsilon$ we have that the substrings $x^{\mathcal{A}}$ and $y^{\mathcal{A}}$ are conditional $2 \epsilon$-typical.
\end{proof}

\end{document}